\documentclass[12pt, a4paper]{article}

\usepackage[pagebackref, letterpaper=true, colorlinks=true, pdfpagemode=noneh, urlcolor=blue, linkcolor=blue, citecolor=blue, pdfstartview=FitH]{hyperref}
\usepackage{amsmath, amsfonts, amssymb}
\usepackage{graphicx}
\usepackage{color}
\usepackage[title]{appendix}
\usepackage[margin=1in]{geometry}

\newif\ifblog
\newif\iftex
\blogfalse
\textrue

\newcommand{\lemref}[1]{Lemma~{\rm \ref{#1}}}
\newcommand{\corref}[1]{Corollary~{\rm \ref{#1}}}

\def\E{{\mathbb E}}

\def\R{{\mathbb R}}

\def\R{{\mathbb R}}

\newcommand{\nd}{\noindent}

\newcommand{\bed}{\begin{displaymath}}
\newcommand{\eed}{\end{displaymath}}
\newcommand{\bea}{\bed\begin{array}{rl}}
\newcommand{\eea}{\end{array}\eed}

\newcommand{\barray}{\begin{array}{ll}}
\newcommand{\earray}{\end{array}}

\newtheorem{theorem}{Theorem}[section]
\newtheorem{lemma}[theorem]{Lemma}

\newtheorem{corollary}[theorem]{Corollary}

\newtheorem{remark}[theorem]{Remark}

\newenvironment{proof}{\noindent {\sc Proof:}}{\strut\hfill $\Box$\newline}

\DeclareMathOperator*{\argmax}{\mathrm{argmax}}

\renewcommand{\geq}{\geqslant}
\renewcommand{\leq}{\leqslant}

\title{A consumption-investment model \\ with state-dependent lower bound 
constraint\\
on consumption}
\date{December 2021}

\author{Chonghu Guan\thanks{School of Mathematics, Jiaying University, Meizhou 514015, Guangdong,
China. This author is partially supported by Guangdong Basic and Applied Basic Foundation (No. 2021A1515012031) and NNSF of China (No. 11901244). Email: {316346917@qq.com}.} 
\and Zuo Quan Xu\thanks{Department of Applied Mathematics, The Hong Kong Polytechnic University, Kowloon, Hong Kong. This author is partially supported by NSFC (No.~11971409), Hong Kong GRF (No.~15204216 and No.~15202817), The PolyU-SDU Joint Research Center on Financial Mathematics and the CAS AMSS-PolyU Joint Laboratory of Applied Mathematics, The Hong Kong Polytechnic University. Email: {maxu@polyu.edu.hk}.}
\and Fahuai Yi\thanks{School of Mathematical Sciences, South China Normal University, Guangzhou, China. This author is partially supported by NNSF of China (No.~12171169). Email: fhyi@scnu.edu.cn. }
}

\begin{document}
\maketitle

\begin{abstract} 
This paper studies a life-time consumption-investment problem under the Black-Scholes framework, where the consumption rate is subject to a lower bound constraint that linearly depends on her wealth. 
It is a stochastic control problem with state-dependent control constraint to which the standard stochastic control theory cannot be directly applied. We overcome this by transforming it into an equivalent stochastic control problem in which the control constraint is state-independent so that the standard theory can be applied. We give an explicit optimal consumption-investment strategy when the constraint is homogeneous. 
When the constraint is non-homogeneous, it is shown that the value function is third-order continuously differentiable by differential equation approach, and a feedback form optimal consumption-investment strategy is provided. According to our findings, if one is concerned with long-term more than short-term consumption, then she should always consume as few as possible; otherwise, she should consume optimally when her wealth is above a threshold, and consume as few as possible when her wealth is below the threshold. 
\bigskip

\nd {\bf Keywords.} Dynamic programming; viscosity solution; optimal consumption-investment; dual transformation; state-dependent constraint

\bigskip

\nd {\bf Mathematics Subject Classification.} 35R35; 91G10; 93E20.

\end{abstract}

\section{Introduction}

The landmark paper \emph{Portfolio Selection} (1952) and book of the same title (1959) by the Nobel laureate Harry M. Markowitz heralded a new era in contemporary finance. His work, however, did not account for the influence of consumption. To address the importance of consumption, Samuelson (1969), Pratt (1964), Arrow (1965), Merton (1969) among many others proposed the optimal consumption-investment models. The goal of these models is to discover the best consumption-investment strategy to maximize the investor's expected utility throughout an investment horizon.
Merton (1975) believed that study of consumption-investment problems is a logical starting point for constructing finance theory. He developed a number of consumption-investment models, including multi-assets with log-normal and more general returns, wage income, and uncertain lifetimes. Merton employed the dynamic programming principle extensively to study these models. This paper will follow his routine. 
\par 
Following Samuelson and Merton's groundbreaking work, many researchers extended the optimal consumption-investment models to incorporate different limitations on trading tactics (see, e.g., Davis and Norman (1990), Fleming and Zariphopoulou (1991), Zariphopoulou (1992), Cvitani and Karatzas (1992, 1993), Zariphopoulou (1994), Shreve and Soner (1994), Akian, Menaldi, and Sulem (1996), Elie and Touzi (2008), Dai and Yi (2009)). 
For a system account, interested readers might refer to Sethi (1997).
\par
Meanwhile, many restrictions on consumption exist in real practice. Investment firms that have cash flow commitments, for instance, are frequently subject to regulatory capital limitations. 
Models with consumption constraints are barely discussed in the literature as compared to models with trading constraints. To integrate such limitations into consumption, Bardhan (1994) explored a model in which the investor must consume no less than a positive rate over the investment horizon. This is referred to as the subsistence consumption requirement. Lakner and Nygren (2006) used martingale method to solve a finite time portfolio selection problem with subsistence consumption requirement and downside constraints on the terminal wealth.
Xu and Yi (2016) looked at a similar model but the consumption rate is constrained by a state-dependent upper bound. The upper bound may be interpreted as fund manager's maximum performance dependent salary. 

\par
In this paper, we consider a consumption-investment model in which the investor's consumption rate is subject to a lower bound constraint that linearly depends on her wealth in a standard Black-Scholes market over a lifetime trading horizon. We make the usual assumption that shorting is allowed but bankruptcy is prohibited. The main tools we use are from stochastic control and ordinary differential equation (ODE, for short) theories (see, e.g., Crandall and Lions (1983), Lions (1983), Fleming and Soner (1992), Dai, Xu and Zhou (2010), Dai and Xu (2011), Chen and Yi (2012), Xu and Yi (2016)). The feedback form optimal consumption-investment strategy is provided. 

As is widely recognized, there is no well-established theory that deals with control problems with state-dependent constraints. The usual maximum principle and forward-backward stochastic differential equation (SDE) method in control theory cannot be directly applied to those problems. We overcome this obstacle by transforming the problem into an equivalent stochastic control problem that does not impose state-dependent constraint on the control variable, allowing the conventional technique to work. This allows us to show the value function is a viscosity solution to the Hamilton-Jacobi-Bellman (HJB) equation, as well as to establish the corresponding verification theorem. When the constraint is homogeneous, we can solve the HJB equation completely and give an explicit optimal strategy. When the constraint is non-homogeneous, we cannot solve the HJB equation, but we can show that the value function is third-order continuously differentiable, implying that a classical solution to the HJB equation exists so that the verification technique works. 

The most important economic contribution of this work is the discovery that the discount factor is critical in determining the optimal consumption habits when the constraint is non-homogeneous. If the discount factor is smaller than some threshold, namely the investor cares about long-term more than short-term consumption, then our result suggests that she should, regardless of her financial situation, always consume as few as possible and save the reminder for future investment, resulting in a higher future expected utility than spending today. By contrast, if she is primarily concerned with short-term consumption when the discount factor is bigger than some threshold, then she should consume optimally when her wealth is above a threshold, and consume as few as possible when her wealth is below the threshold. 

\par
The remainder of the paper is organized as follows. In Section 2, we formulate a consumption-investment problem with state-dependent constraint and resolve its feasibility issue. In Section 3, we study the homogeneous constraint case and give an explicit optimal trading strategy. Sections 4-7 focus on the non-homogeneous constraint case. 
Section 4 presents the related HJB equation and verification theorem. 
Section 5 studies the properties of the value function. In Section 6, we first perform a dual transformation to convert the {fully non-linear} HJB equation into a {semi-linear} one, and then use ODE techniques to prove that the solution to the dual equation is three times continuously differentiable, resulting in the same regularity of the value function. Section 7 discusses the optimal consumption-investment strategy and provide our financial findings. Section 8 concludes the paper and the solution with state-independent constraint is provided in Appendix A.

\setcounter{equation}{0}

\section{Problem formulation}
\setcounter{equation}{0}

We consider a small investor (``She'') who has a positive initial endowment $x$ and is trading in a complete, arbitrage-free, continuous-time Black-Scholes financial market. The market consists of two financial assets only: a risk-free asset with instantaneous interest rate $r$ and a risky asset whose price $S$ is governed by a stochastic differential equation: 
\begin{align}\label{eq:S}
\frac{dS_t}{S_t}=(r+\mu)dt+\sigma dW_t,
\end{align}
where $W$ is a standard Brownian motion and $r\geq 0$, $\mu>0$ and $\sigma>0$ are constants.
\par
An investment strategy $\pi_t$ represents the holding amount (in dollars) of the risky asset in the portfolio at time $t$ and a consumption strategy $c_t\geq 0$ represents the rate of consumption. We assume the trading and consumption are self-financing so that the wealth $X_t$ of the investor evolves according to the following SDE:
\begin{align}\label{X}
dX_t=(rX_t-c_t+\mu\pi_t)dt+\pi_t\sigma dW_t. 
\end{align} 
\par
We call a progressively measurable process $(c_t, \pi_t)$ an admissible consumption-investment strategy/control/pair if it satisfies, for all $t\geqslant 0$, 
\begin{align}\label{constraint1}
\E\bigg[\int_0^t(c_s+\pi_s^2)ds\bigg]<\infty, 
\end{align}
and 
\begin{align}\label{constraint2}
X_t\geqslant 0, \quad c_t\geqslant kX_t+l, 
\end{align}
where $k\geq 0$ and $l\geq 0$ are given constants. 
Given any admissible strategy $(c_t, \pi_t)$, the SDE \eqref{X} admits a unique integrable strong solution $X_t$.

\par
The investor's goal is to discover the best consumption-investment $(c_t, \pi_t)$ strategy to maximize her lifetime expected discounted utility: 
\begin{align}\label{V}
V(x):=\sup\limits_{\pi, c}\E\bigg[\int_0^\infty e^{-\beta t}U(c_t)dt\;\Big|\;X_0=x\bigg], 
\end{align}
where $\beta>0$ is a discount factor to ensure the above is finite. Following Merton's (1971) model, we take 
\begin{align*}
U(c)=\frac{c^p}{p}, \quad c\geqslant0,
\end{align*}
where $0<p<1$ is a constant.
Here $V$ is called the value function of the problem \eqref{V}. 
Economically speaking, a small $\beta$ means that the investor cares more about long-term than short-term consumption, and a large $\beta$ means the opposite. Later we will show it is critical in determining optimal consumption habits. 
\par 
Because the control constraint in \eqref{constraint2} is state-dependent, some well-known and powerful methods in control theory, such as the maximum principle and the forward-backward SDE approach, cannot be directly applied to solving the problem \eqref{V}. Our idea is to transform the problem into an equivalent stochastic control problem in which the control constraint is state-independent, allowing the conventional technique to work. 

\subsection{Feasibility issue of the problem \eqref{V}}
\par
We start with the feasibility issue of the problem \eqref{V}, that is, whether it has at least one admissible strategy. 
If the problem is not feasible, then there is nothing to investigate. 

Let 
\[\gamma_t=e^{(k-r-\frac{1}{2}(\sigma^{-1}\mu)^2)t-\sigma^{-1}\mu W_t}.\]
By It\^{o}'s lemma, \eqref{X} and \eqref{constraint2},
\[d \gamma_t=\gamma_t\big((k-r)dt-\sigma^{-1}\mu dW_t\big)\]
and 
\begin{align*}
d (\gamma_t X_t)&=\gamma_tdX_t+X_td\gamma_t-\sigma^{-1}\mu\gamma_t\pi_t\sigma dt\\
&=\gamma_t\big((rX_t-c_t+\mu\pi_t)dt+\pi_t\sigma dW_t\big)
+X_t\gamma_t\big((k-r)dt-\sigma^{-1}\mu dW_t\big)-\sigma^{-1}\mu\gamma_t\pi_t\sigma dt\\
&=\gamma_t\big((kX_t-c_t)dt+\pi_t\sigma dW_t\big)\\
&\leq \gamma_t\big(-l dt+\pi_t\sigma dW_t\big)
\end{align*} 
for any admissible strategy $(c_t, \pi_t)$. Since $l\geq 0$, the above implies $\gamma_t X_t$ is a local suppermartingale. It is also nonnegative, so it is a suppermartingale. 
A localizing argument leads to 
\begin{align}\label{xgrowth} 
0\leq \E\big(\gamma_{t} X_{t}\big) & \leq x- l\E\big(\int_0^{t} \gamma_{s} ds\big)=x-l\int_0^{t} e^{(k-r)s} ds
\end{align} 
for any $t\geq 0$. 
So by the monotone convergent theorem, we must have 
\begin{align}\label{constraint4} 
x\geq l\int_0^{\infty} e^{(k-r)s} ds. 
\end{align} 
Therefore, \eqref{constraint4} is a necessary condition for the problem \eqref{V} being feasible. On the other hand, if \eqref{constraint4} holds, then one can check that $(c_t, \pi_t)=(kX_t+l, 0)$ is always an admissible strategy. 
Therefore, we conclude that the problem \eqref{V} is feasible if and only if \eqref{constraint4} holds.

If $k\geq r$, $l>0$, then \eqref{constraint4} cannot hold as its right hand side is $\infty$, so there is no admissible strategy for the problem \eqref{V}. 

We only need to deal with the left two cases: the homogeneous case $l=0$ and the non-homogeneous case $0\leq k<r$, $l>0$. In these two cases we adopt a viscosity solution plus verification theorem approach to tackle the problem \eqref{V}.

In the rest of this paper, we denote 
\[ \kappa:=\frac{\beta-p\left(\frac{\mu^2}{2\sigma^2(1-p)}+r\right)}{1-p}. \]

\begin{remark}
We point here that the value of $\kappa$ relays on the discount factor $\beta$. Later, we will see that it plays an important role in determining the optimal consumption-investment strategy.
\end{remark}

\section{Solution in the homogeneous case $l=0$}
\setcounter{equation}{0}

We start with the homogeneous case $l=0$. In this case, we give an explicit optimal control for the problem \eqref{V}.
The result will play as a benchmark for the non-homogeneous case. In this case, the constraint \eqref{constraint4} is satisfied for any $x\geq 0$, so the problem \eqref{V} is feasible for any $x\geq 0$.

We first recall Merton's (1971) well-known result for the unconstrained case $k=l=0$. 
\begin{lemma}[Merton's Theorem]\label{Mertons}
If $k=l=0$, $\kappa>0$, then the optimal strategy to the problem \eqref{V} is 
\[ (c_t, \pi_t)=\bigg(\kappa X_t, \frac{\mu}{\sigma^2(1-p)}X_t\bigg), \quad t\geqslant 0, \]
and the optimal value $V(x)$ is equal to 
\[ V_0(x)=\frac{1}{p}\kappa^{p-1}x^p. \]
If $k=l=0$, $\kappa\leq 0$, then the optimal value is infinity. 
\end{lemma} 
\begin{proof}
The claim follows Merton (1971) when $k=l=0$, $\kappa>0$. 
The optimal value of the problem \eqref{V} is clearly decreasing in $\beta$. By the above result the optimal value when $k=l=0$, $\kappa\leq 0$ is at least $\lim_{\kappa\downarrow 0}\frac{1}{p}\kappa^{p-1}x^p=\infty$, so the optimal value of the problem is infinity. 
\end{proof}

It is clear that $V(x)\leq V_0(x)$ in all cases because the admissible strategy set for the latter is bigger than that of the former. 
Based on this, we can show 
\begin{lemma}[Viscosity solution in homogeneous case]
\label{lemmaHJB-2}
If $l=0$, $k>0$, $\kappa>0$, then the function $V$ defined by \eqref{V} is a viscosity solution to the following HJB equation
\begin{align}\label{HJB-2}
\begin{cases}
\beta V-\sup\limits_{\pi}\bigg(\frac{1}{2}\sigma^2\pi^2V_{xx}+\mu\pi
V_x\bigg)-\sup\limits_{c\geqslant kx}\bigg(\frac{c^p}{p}-cV_x\bigg)-rxV_x=0, \quad x>0,\\
V(0)=0,
\end{cases}
\end{align}
and satisfies $0\leq V(x)\leq C(x^p+1)$ for some constant $C>0$. 
\end{lemma}
The standard approach for viscosity solution is inapplicable here, because the approach requires that the control constraint set does not relay on the controlled state process (see, e.g. Yong and Zhou (1999)). 
The state-dependent constraint on consumption rate prevents us from freely selecting the control variables when checking the viscosity solution. This is the critical difference between our problem and standard theory. We overcome this obstacle by transforming the problem into an equivalent stochastic control problem that does not impose state-dependent constraints on control variables, allowing the conventional theory to apply. The proof of Lemma \ref{lemmaHJB-2} is given as follows. \medskip

\begin{proof}
The power growth estimate is due to $V(x)\leq V_0(x)$. 
We introduce a new control variable $u_t=c_t-kX_t$. When $l=0$, one can see the constraints \eqref{constraint2} and \eqref{constraint1} are equivalent to, for all $t\geqslant 0$, 
\begin{align*}
\E\bigg[\int_0^t(u_s+\pi_s^2)ds\bigg]<\infty, \quad
X_t\geqslant 0, \quad u_t\geqslant 0. 
\end{align*}
As a consequence, 
\begin{align*} 
V(x)=\sup\limits_{\pi, u}\E\bigg[\int_0^\infty e^{-\beta t}U(u_t+kX_t)dt\;\Big|\;X_0=x\bigg]. 
\end{align*}
Because the restriction on the new control variable $u_t$ is independent of the state process $X_t$, the standard approach can be applied to prove that $V$ is a viscosity solution to the following equation
\begin{align}\label{HJB-3}
\begin{cases}
\beta V-\sup\limits_{\pi}\bigg(\frac{1}{2}\sigma^2\pi^2V_{xx}+\mu\pi
V_x\bigg)-\sup\limits_{u\geq 0}\bigg(\frac{(u+ kx)^p}{p}-(u+ kx)V_x\bigg)-rxV_x=0, \quad x>0,\\
V(0)=0.
\end{cases}
\end{align}
We omit the proof here as it is standard. Interested readers may refer to Crandall and Lions (1983), Lions (1983), Fleming and Soner (1992) and Yong and Zhou (1999) for details. 
The above ODE is clearly equivalent to the desired HJB equation \eqref{HJB-2}, so the proof is complete. 
\end{proof} 

In a standard viscosity approach, one also needs to show that the HJB equation \eqref{HJB-2} admits at most one viscosity solution, which thus must be the value function $V$ by Lemma \ref{lemmaHJB-2}. This approach, however, is generally quiet complex and difficult, and usually cannot show the existence of a classical $C^2$ solution to the HJB equation. In this paper we adopt a different verification approach. It asserts that a classical $C^2$ solution to the HJB equation must be the value function. 

\begin{lemma}[Verification theorem in homogeneous case]\label{varthm-2}
Suppose $l=0$, $k>0$, $\kappa>0$. Then $\varphi=V$ if $\varphi\in C^2((0,+\infty))$ is a classical solution to the HJB equation \eqref{HJB-2} and satisfies $0\leq \varphi(x)\leq C(x^p+1)$ for some constant $C>0$. 
\end{lemma}
\begin{proof}
Again, because of the state-dependent constraint on the control variables, the standard approach to proving the verification theorem fails.
However, we can still utilize the approach used to prove Lemma \ref{lemmaHJB-2} to prove this statement.
We use the same notation as in the proof of Lemma \ref{lemmaHJB-2}. 
By \eqref{HJB-3}, we have, for any $x>0$, $u\geq 0$, $\pi\in\R$, 
\[ \beta \varphi\geq \frac{1}{2}\sigma^2\pi^2 \varphi_{xx}+\mu\pi\varphi_x+U(u+ kx) -(u+ kx)\varphi_x+rx\varphi_x.\]
Applying It\^{o}'s lemma to $e^{-\beta t}\varphi(X_t)$, 
\begin{align} 
e^{-\beta t}\varphi(X_t) &=\varphi(x)+\int_0^t e^{-\beta s}\Big(-\beta \varphi(X_s)
+ \varphi_x(X_s)((rX_s-(u_s+kX_s)+\mu\pi_s)\nonumber\\
&\qquad\qquad\qquad\qquad\;\;+\frac{1}{2}\varphi_{xx}(X_s) \pi_s^2\sigma^2
\Big)ds+\int_0^t e^{-\beta s}\varphi_x(X_s)\pi_s\sigma dW_s\nonumber\\
&\leq \varphi(x)-\int_0^t e^{-\beta s}U(u_s+kX_s)ds+\int_0^t e^{-\beta s}\varphi_x(X_s)\pi_s\sigma dW_s.\label{veriphi}
\end{align}
Let $\theta_n=n\wedge \inf\{t\geq 0: |\varphi_x(X_t)\pi_t|>n\}$. Then
\begin{align} \label{veriphi2}
\E\bigg(\int_0^{ \theta_n} e^{-\beta s}U(u_s+kX_s)ds\bigg)
&\leq \varphi(x)-\E\big(e^{-\beta \theta_n}\varphi(X_{ \theta_n})\big) .
\end{align}
Applying the monotone convergence theorem, we obtain 
\begin{align} \label{veriphi2-1}
\E\bigg(\int_0^{\infty} e^{-\beta s}U(u_s+kX_s)ds\bigg)&\leq \varphi(x)-\liminf_n\E\big(e^{-\beta \theta_n}\varphi(X_{ \theta_n})\big).
\end{align}
Since $\varphi\geq 0$, it yields 
\begin{align*} 
\E\bigg(\int_0^{\infty} e^{-\beta s}U(u_s+kX_s)ds\bigg)&\leq \varphi(x), 
\end{align*}
which implies $\varphi(x)\geq V(x)$.

On the other hand, if we chose a feedback pair $(u, \pi)$ to maximize the optimization problem in \eqref{HJB-3}, then the inequalities in \eqref{veriphi}-\eqref{veriphi2-1} become equations.
If we can show 
\begin{align} \label{veriphi3}
\liminf_n\E\big(e^{-\beta \theta_n}\varphi(X_{ \theta_n})\big)=0,
\end{align}
then 
\begin{align*} 
\E\bigg(\int_0^{\infty} e^{-\beta s}U(u_s+kX_s)ds\bigg)&= \varphi(x), 
\end{align*}
implying $V(x)\geq \varphi(x)$. This will complete the proof of the lemma. 

It is only left to prove \eqref{veriphi3}. Using H\"older's inequality 
$\E(XY)\leq [\E(X^{\frac{1}{p}})]^p [\E(Y^{\frac{1}{1-p}})]^{1-p}$ for $X,Y\geq 0$ and the fact that $\gamma_t X_t$ is a suppermartingale, we get 
\begin{align*} 
\E\big(e^{-\beta \theta_n}\varphi(X_{ \theta_n})\big) 
&\leq \E\big(e^{-\beta \theta_n} C((X_{ \theta_n})^p+1)\big) \\
&=C \E\big(e^{-\beta \theta_n} \gamma_{\theta_n}^{-p} (\gamma_{\theta_n} X_{ \theta_n})^p\big) +C\E\big(e^{-\beta \theta_n}\big)\\
&\leq C\big[\E\big(\gamma_{\theta_n} X_{ \theta_n}\big)\big]^p
\big[\E\big((e^{-\beta \theta_n} \gamma_{\theta_n}^{-p})^{\frac{1}{1-p}}\big)\big]^{1-p}+C\E\big(e^{-\beta \theta_n}\big)\\
&\leq Cx^p
\big[\E\big((e^{-\beta \theta_n} \gamma_{\theta_n}^{-p})^{\frac{1}{1-p}}\big)\big]^{1-p}+C\E\big(e^{-\beta \theta_n}\big).
\end{align*}
By H\"older's inequality again, 
\begin{align*} 
\E\big((e^{-\beta \theta_n} \gamma_{\theta_n}^{-p})^{\frac{1}{1-p}}\big) &=
\E\big(e^{-\frac{\beta+p(k-r-\frac{1}{2}(\sigma^{-1}\mu)^2)}{1-p}\theta_n+\frac{p\sigma^{-1}\mu}{1-p} W_{\theta_n}}\big)\\
&=\E\big(e^{-(\kappa+\frac{pk}{1-p}+\frac{1}{2}a^2)\theta_n+aW_{\theta_n}}\big)\\
&\leq\big[\E\big(e^{-\frac{1}{\eta}(\kappa+\frac{pk}{1-p}+\frac{1}{2}a^2-\frac{a^2}{2(1-\eta)})\theta_n}\big)\big]^{\eta}\big[\E(e^{-\frac{a^2}{2(1-\eta)^2}\theta_n+\frac{a}{1-\eta}W_{\theta_n}})\big]^{1-\eta}, 
\end{align*} 
where $a=\frac{p\sigma^{-1}\mu}{1-p}$ and $0<\eta<1$.
Since $k\geq 0$, $\kappa>0$, we can let $\eta$ be sufficiently close to 0 so that 
\[\kappa+\frac{pk}{1-p}+\frac{1}{2}a^2-\frac{a^2}{2(1-\eta)}>0.\]
By the monotone convergence theorem, it follows 
\[\lim_n\E\big(e^{-\frac{1}{\eta}(\kappa+\frac{pk}{1-p}+\frac{1}{2}a^2-\frac{a^2}{2(1-\eta)})\theta_n}\big)=0.\]
It is easily seen from a localizing argument that 
\[\E(e^{-\frac{a^2}{2(1-\eta)^2}\theta_n+\frac{a}{1-\eta} W_{\theta_n}})\leq 1.\]
Combining the above estimates, we establish \eqref{veriphi3}. 
\end{proof}

The following result gives the optimal strategy and optimal value in the homogeneous case. 
\begin{theorem}[Optimal strategy in homogeneous case]\label{homogeneous}
If $l=0$, $\kappa>0$, then the optimal consumption-investment strategy to the problem \eqref{V} is
\[ (c_t, \pi_t)=\bigg(\max\{\kappa, k\} X_t, \; \frac{\mu}{\sigma^2(1-p)}X_t\bigg), \quad t\geqslant 0, \]
and the optimal value $V(x)$ is equal to 
\begin{align}\label{Vk}
V_k(x)=\frac{\max\{\kappa, k\}^p}{p(\kappa(1-p)+\max\{\kappa, k\}p)}x^p=
\left\{
\begin{array}{ll}
\frac{k^p}{p(\kappa(1-p)+kp)}x^p, &\text{if } k>\kappa; \\[10pt]
\frac{1}{p}\kappa^{p-1}x^p, & \text{if } 0\leq k\leqslant \kappa.
\end{array}
\right.
\end{align}
\end{theorem}
\begin{proof} 
When $k=0$, this result is Merton's Theorem \ref{Mertons}. 
When $k>0$, the assertion follows from Lemma \ref{varthm-2}.
\end{proof}
\par
As the admissible set of $c_t$ is decreasing in $k$, we see $V_k(x)$ is decreasing in $k$, and consequently $V_k(x)\leqslant V_0(x)=\frac{1}{p}\kappa^{p-1}x^p$. However, when $k\leqslant \kappa$, the constraint $c_{t}\geqslant kX_{t}$ is satisfied by the optimal consumption rate $c_{t}=\kappa X_{t}$ in Merton's Theorem, so the lower bound constraint becomes redundant. By contrast, when $k>\kappa$, the Merton strategy is no more feasible, but the optimal consumption strategy is till proportional to the wealth.

Now we have completely solved the problem \eqref{V} in the homogeneous case. The rest of this paper will focus on the non-homogeneous case.

\section{The non-homogeneous case $0\leq k<r$, $l>0$: HJB equation and verification theorem }
\setcounter{equation}{0}

We start with the feasibility issue of the problem \eqref{V} in the non-homogeneous case $0\leq k<r$, $l>0$. 

In this case, the feasibility condition \eqref{constraint4} becomes 
$x\geq \frac{l}{r-k}$, so the problem \eqref{V} is feasible if and only if $x\geq \frac{l}{r-k}$. 
In the marginal case $x=\frac{l}{r-k}$, $(c_t, \pi_t)=(kX_t+l, 0)$ is the only admissible strategy and the corresponding state process is $X_t\equiv \frac{l}{r-k}$.
Hence, we deduce a left boundary condition for the value function $V$: 
\begin{align}\label{V00}
V\left(\frac{l}{r-k}\right)=\int_0^\infty e^{-\beta t}\frac{(k\frac{l}{r-k}+l)^p}{p}dt=\frac{1}{\beta p}\left(\frac{rl}{r-k}\right)^p.
\end{align}
In the rest of this paper, we denote
\[ x_e=\frac{l}{r-k}> 0, \quad c_e=kx_e+l=\frac{rl}{r-k}> 0. \]
The left boundary condition \eqref{V00} can be expressed as
\begin{align}\label{V0}
V(x_e)=\frac{1}{\beta p}c_e^p.
\end{align}
As discussed above, the problem \eqref{V} has no admissible solution when $x<x_e$. 

Similar to the homogeneous case, we have the following two results. Their proofs are similar to the homogeneous case, so we omit the details. 
\begin{lemma}[Viscosity solution in non-homogeneous case]
\label{lemmaHJB}
If $0\leq k<r$, $l>0$, $\kappa>0$, then the function $V$ defined by \eqref{V} is a viscosity solution to the following HJB equation
\begin{align}\label{HJB}
\begin{cases}
\beta V-\sup\limits_{\pi}\bigg(\frac{1}{2}\sigma^2\pi^2V_{xx}+\mu\pi
V_x\bigg)-\sup\limits_{c\geqslant kx+l}\bigg(\frac{c^p}{p}-cV_x\bigg)-rxV_x=0, \quad x>x_e,\\
V(x_e)=\frac{1}{\beta p}c_e^p,
\end{cases}
\end{align}
and satisfies $0\leq V(x)\leq C(x^p+1)$ for some constant $C>0$. 
\end{lemma}

\begin{lemma}[Verification theorem in non-homogeneous case]\label{varthm}
Suppose $0\leq k<r$, $l>0$, $\kappa>0$. Then $\varphi=V$ if $\varphi\in C^2((x_e,+\infty))$ is a classical solution to the HJB equation \eqref{HJB} and satisfies $0\leq \varphi(x)\leq C(x^p+1)$ for some constant $C>0$. 
\end{lemma}

\par
By Lemmas \ref{lemmaHJB} and \ref{varthm}, solving the problem \eqref{V} reduces to finding a $C^2$ solution to the HJB equation \eqref{HJB}. Different from the homogeneous case, we can not get an explicit solution in the non-homogeneous case, except the case with state-independent constraint (in which case the solution is given in Appendix A). 
We will first use a dual transformation to convert the {fully non-linear} HJB equation \eqref{HJB} into a {semi-linear} one and then show the latter admits a classical $C^3$ solution, resulting in a $C^3$ solution to the HJB equation \eqref{HJB}. 
\par


\section{Properties of the value function in non-homogeneous case }
\setcounter{equation}{0} 
We focus on the case $k>0$, $l>0$, $\kappa>0$. We give an explicit solution in the special case $k=0$, $l>0$, $\kappa>0$, in Appendix A. 

\subsection{Bounds for the value function}
\setcounter{equation}{0}

\begin{theorem} We have
\begin{align}\label{V_up}
V(x)\leqslant V_k(x),
\end{align}
in $(x_{e}, \infty)$, where $V_k(x)$ is defined in Theorem \ref{homogeneous}.
\end{theorem}

\begin{proof}
As the admissible set of $c_t$ is decreasing in $l$ and thus the value function corresponding to $l=0$ gives an upper bound for $V(x)$.
\end{proof}

To gain a better understanding of the value function, we will explore a new problem and utilize it to deduce the properties for $V(x)$.
Let
\[ Y_t=X_t-x_e, \quad \xi_t=c_t-c_e, \]
then
\begin{align}\label{dY}
dY_t 
&=[rY_t-\xi_t+\mu\pi_t]dt+\sigma\pi_t dW_t,\\\label{Y0}
Y_0&=y=x-x_e,
\end{align}
and by \eqref{constraint2}, we have the constraint 
\begin{align}\label{Y_xi}
Y_t\geqslant 0, \quad \xi_t\geqslant kY_t.
\end{align}
Define
\begin{align}\label{W}
W(y):=\sup\limits_{\pi, \xi}\E\bigg[\int_0^\infty e^{-\beta t}U(\xi_t+c_e)dt\bigg],\quad y>0.
\end{align}
Then
\begin{align}\label{VW}
V(x)=W(x-x_e),\quad x>x_{e}.
\end{align}

Now we can prove the following bounds for $V$. 
\begin{theorem}\label{Thm:V_down}
We have
\begin{align}\label{V_downup}
V_k(x-x_e)\leqslant V(x)\leqslant V_k(x-x_e)+V(x_e), \quad x> x_{e}. 
\end{align}
\end{theorem}
\begin{proof}
By the definitions of $W(y)$ and $V_k(x)$, we have
\begin{align}\label{W_down}
W(y)=\sup\limits_{\pi, \xi}\E\bigg[\int_0^\infty e^{-\beta
t}U(\xi_t+c_e)dt\bigg] \geqslant \sup\limits_{\pi, \xi}\E\bigg[\int_0^\infty
e^{-\beta t}U(\xi_t)dt\bigg]=V_k(y),
\end{align}
so the lower bound follows from \eqref{VW}. 
\par
On the other hand, recall that $U(c)=\frac{c^p}{p}$, so $U(x+y)\leqslant U(x)+U(y)$ for $x,y\geqslant 0$. Thus
\begin{align*}
W(y)&=\sup\limits_{\pi, \xi}\E\bigg[\int_0^\infty e^{-\beta
t}U(\xi_t+c_e)dt\bigg]\\
&\leqslant \sup\limits_{\pi, \xi}\E\bigg[\int_0^\infty
e^{-\beta t}U(\xi_t)dt\bigg]+\sup\limits_{\pi, \xi}\E\bigg[\int_0^\infty
e^{-\beta t}U(c_e)dt\bigg]=
V_k(y)+V(x_e),
\end{align*}
thanks to \eqref{V0}.
And consequently, the upper bound follows from \eqref{VW}. 
\end{proof}

Based on this result, we can establish the continuity of $V$. 
\begin{corollary}
The value function $V(x)$ is continuous on $[x_e, +\infty)$.
\end{corollary}
\begin{proof}
The concavity of $V(x)$ implies its continuity on $(x_e, +\infty)$. The monotonicity of $V(x)$ implies $V(x_e)\leqslant V(x_e+)$; whereas the upper bound in \eqref{V_downup} indicates $V(x_e+)\leqslant V(x_e)$ as $V_{k}(0+)=0$. Hence, we established the continuity of $V(x)$ on $[x_e, +\infty)$.
\end{proof} 

In the following section, we study the continuity of the derivative of $V$.


\subsection{Properties of the first-order derivative}

\begin{theorem}
The first-order derivative function $V_x(x)$ exists and is continuous on $(x_e, +\infty)$.
\end{theorem}
\begin{proof}
As $V(x)$ is concave, the right and left derivatives
\[ V_x(x\pm):=\lim\limits_{y\rightarrow x\pm}\frac{V(y)-V(x)}{y-x} \]
are well-defined. By the monotonicity and concavity of $V(x)$, $V_x(x\pm)$ are both decreasing and satisfy 
$0\leqslant V_x(x+)\leqslant V_x(x-)<+\infty$. To prove the continuity of $V_{x}(x)$, by Darboux's theorem, it suffices to prove that
$V_x(x+)=V_x(x-)$.
\par
We argue by contradiction. Suppose that there exists $x_0>x_e$, such that
$V_x(x_0+)<V_x(x_0-)$. Then, for any $\theta\in (V_x(x_0+), V_x(x_0-))$ and $N>0$, there exists
$\varepsilon>0$ that
\[ V(x)\leqslant V(x_0)+\theta(x-x_0)-N(x-x_0)^2, \quad x\in(x_0-\varepsilon, x_0+\varepsilon). \]
Let $\phi(x)=V(x_0)+\theta(x-x_0)-N(x-x_0)^2$, then $V(x)-\phi(x)$ attains its local maximum value 0 at the point $x_0$. Because $V(x)$ is the viscosity solution to \eqref{HJB} and $\phi(x)\in C^2$, 
\begin{align*}
\beta \phi(x_0)-\sup\limits_{\pi}\bigg(\frac{1}{2}\sigma^2\pi^2\phi_{xx}(x_0)+\mu\pi \phi_x(x_0)\bigg)-\sup\limits_{c\geqslant kx_0+l}\bigg(\frac{c^p}{p}-c\phi_x(x_0)\bigg)-rx\phi_x(x_0)\leqslant 0,
\end{align*}
and consequently,
\begin{align*}
\beta V(x_0)=\beta \phi(x_0) &\leqslant \sup\limits_{\pi}\bigg(\frac{1}{2}\sigma^2\pi^2(-2N)+\mu\pi
\theta\bigg)+\sup\limits_{c\geqslant kx_0+l}\bigg(\frac{c^p}{p}-c\theta\bigg)+rx_0\theta, \\
&= \frac{\mu^2\theta^2}{4\sigma^2N}+\sup\limits_{c\geqslant kx_0+l}\bigg(\frac{c^p}{p}-c\theta\bigg)+rx_0\theta.
\end{align*}
Denoting $g(x,\theta)=\sup\limits_{c\geqslant kx+l}\bigg(\frac{c^p}{p}-c\theta\bigg)+rx\theta$ and letting $N\rightarrow+\infty$ in the above inequality, we obtain
\begin{align}\label{V<g}
\beta V(x_0)\leqslant g(x_{0},\theta), \quad \theta\in(V_x(x_0+), V_x(x_0-)).
\end{align}
\par
On the other hand, due to the concavity of $V(x)$, $V_x(x)$ is continuous and differentiable
almost everywhere, so there exists a sequence $\{x_n\}$ going up to $x_0$ such that $V_{xx}(x_n)$ exists for all $n$. Because $V(x)$ is the viscosity solution to \eqref{HJB}, we see 
\begin{align*}
\beta V(x_n)
&=\sup\limits_{\pi}\bigg(\frac{1}{2}\sigma^2\pi^2V_{xx}(x_n)+\mu\pi V_x(x_n)\bigg)+\sup\limits_{c\geqslant kx_n+l}\bigg(\frac{c^p}{p}-cV_x(x_n)\bigg)-rx_nV_x(x_n)\\
&\geqslant \sup\limits_{c\geqslant kx_n+l}\bigg(\frac{c^p}{p}-cV_x(x_n)\bigg)-rx_nV_x(x_n)\\
&= g(x_{n},V_x(x_n))
\end{align*}
Because $V(x)$ and $g(x,\theta)$ are both continuous, we deduce
\begin{align*}
\beta V(x_0)\geqslant g(x_{0},V_x(x_0-)).
\end{align*}
Similarly, by choosing a decreasing sequence, we can prove 
\begin{align*}
\beta V(x_0)\geqslant g(x_{0},V_x(x_0+)).
\end{align*}
Note that the mapping $\theta\mapsto g(x_{0},\theta)$ is convex, thus
\[ \beta V(x_0)\geqslant \max\{g(x_{0},V_x(x_0+)), g(x_{0},V_x(x_0-))\}\geqslant g(x_{0},\theta), \quad \theta\in(V_x(x_0+), V_x(x_0-)). \]
Together with \eqref{V<g} we infer that $\theta\mapsto g(x_{0},\theta)$ is a constant funcion on $(V_x(x_0+), V_x(x_0-))$. By the definition of $g(x,\theta)$, we have
\begin{align*}
g(x_{0},\theta)=
\left\{
\begin{array}{ll}
\frac{1-p}{p}\theta^{\frac{p}{p-1}}+rx_0\theta, & \hbox{if } {\theta}^{\frac{1}{p-1}}\geqslant kx_0+l;\\
\frac{(kx_0+l)^p}{p}+[(r-k)x_0-l]\theta, & \hbox{if } {\theta}^{\frac{1}{p-1}}\leqslant kx_0+l.
\end{array}
\right.
\end{align*}
Consequently, $\theta\mapsto g(x_{0},\theta)$ cannot be a constant function as $x_0>x_e$ (which implies $(r-k)x_0-l>(r-k)x_e-l=0$), leading to a contradiction. 
\end{proof}


We next give some bounds for the first-order derivative function. 
\begin{lemma}\label{Lem:Vx0} We have
\[ V_x(x_e+):=\lim\limits_{x\rightarrow x_e+}V_x(x)=+\infty. \]
\end{lemma}
\begin{proof}
We come to prove the following equivalent conclusion
\[ W_y(0+)=+\infty, \]
where $W(y)$ is defined in \eqref{W}.
\par
Fix $h>\beta/\mu$. We choose the feedback controls
\begin{align*}
\pi_t=h Y_t, \quad
\xi_t=r Y_t,
\end{align*}
in \eqref{dY}, then $Y_t$ is the solution to the following SDE,
\begin{align*} 
\left\{
\begin{array}{l}
dY_t=
\mu h Y_t dt+\sigma h Y_t dW_t, \quad t\geqslant 0, \\
Y_{0}=y,
\end{array}
\right.
\end{align*}
or equivalently,
\begin{align*}
Y_t=ye^{(\mu h-\frac{\sigma^2 h^2}{2})t+\sigma h W_t}.
\end{align*}
As a consequence 
\begin{align*}
\xi_t=r Y_t=rye^{(\mu h-\frac{\sigma^2 h^2}{2})t+\sigma h W_t}.
\end{align*}
By \eqref{W}, we have
\begin{align*}
W(y)\geqslant \E\bigg[\int_0^\infty e^{-\beta t}\frac{1}{p}\bigg(rye^{(\mu h-\frac{\sigma^2 h^2}{2})t+\sigma h W_t}+c_e\bigg)^pdt\bigg].
\end{align*}
Thus, Fatou's lemma yields
\begin{align*}
\liminf\limits_{y\rightarrow 0+}\frac{W(y)-W(0)}{y}
&\geqslant
\liminf\limits_{y\rightarrow 0+}\frac{1}{y}\E\bigg[\int_0^\infty e^{-\beta
t}\frac{1}{p}\bigg(\bigg(rye^{(\mu h-\frac{\sigma^2 h^2}{2})t+\sigma h W_t}+c_e\bigg)^p-c_e^p\bigg)dt\bigg]\\
&\geqslant
\E\bigg[\int_0^\infty e^{-\beta
t}\frac{1}{p}\liminf\limits_{y\rightarrow 0+}\frac{\bigg(rye^{(\mu h-\frac{\sigma^2 h^2}{2})t
+\sigma h W_t}+c_e\bigg)^p-c_e^p}{y}dt\bigg]\\
&= \E\bigg[\int_0^\infty e^{-\beta t}c_e^{p-1}re^{(\mu h-\frac{\sigma^2 h^2}{2})t +\sigma h W_t}dt\bigg]\\
&= c_e^{p-1}r\int_0^\infty e^{-\beta t}e^{\mu ht} dt\\
&=+\infty.
\end{align*}
Therefore, we conclude from the concavity of $W(y)$ that 
\begin{align*}
\liminf\limits_{y\rightarrow0+}W_y(y)\geqslant \liminf\limits_{y\rightarrow 0+}\frac{W(y)-W(0)}{y}=+\infty.
\end{align*}
This completes the proof.
\end{proof}
\par
We now give an upper bound for the first-order derivative function.
\begin{lemma}\label{Lem:Vx_up}
We have
\begin{align}\label{Vx_up}
V_x(x)\leqslant \frac{\max\{\kappa,k\}^{p}x^p}{\kappa(x-x_e)},\quad x>x_e,
\end{align}
and hence $\lim\limits_{x\rightarrow+\infty} V_{x}(x)=0$.
\end{lemma}
\begin{proof}
If we can show 
\begin{align}\label{Wy}
yW_y(y)\leqslant pW(y),\quad y>0, 
\end{align}
then from \eqref{VW}, \eqref{V_up} and \eqref{Vk}, 
\[ V_x(x)=W_y(x-x_e)\leqslant \frac{pW(x-x_e)}{x-x_e}=\frac{pV(x)}{x-x_e}\leqslant \frac{pV_k(x)}{x-x_e}\leqslant \frac{\max\{\kappa,k\}^{p}x^p}{\kappa(x-x_e)}, \]
the desired result follows. 
To show \eqref{Wy}, let $W(y, c_e)$ be the value function defined by \eqref{W}. From \eqref{dY}, \eqref{Y0} and \eqref{Y_xi}, we see that $W(\cdot, \cdot)$ is homogeneous of degree $p$, namely
\[ W(\lambda y, \lambda c_e)=\lambda^pW(y, c_e), \quad \lambda\geqslant 0. \]
Let $\lambda=y^{-1}$,
\[ W(1, y^{-1}c_e)=y^{-p}W(y, c_e). \]
Clearly $W(1, y^{-1}c_e)$ is decreasing in $y$, so 
\[ 0\geqslant \partial_y[y^{-p}W(y, c_e)]=\frac{W_y(y, c_e)y^p-pW(y, c_e)y^{p-1}}{y^{2p}}, \]
giving \eqref{Wy}.
\end{proof}
\begin{lemma}\label{Lem:Vxx}
We have $V_x(x)>0$ in $(x_e, +\infty)$, and $V_{xx}(x)<0$ almost everywhere in $(x_e, +\infty)$.
\end{lemma}
\begin{proof}
Because $V$ is concave and strictly increasing, it has no stationary point, resulting in $V_x>0$. 
The concavity of $V(x)$ implies that $V_{xx}(x)\leq 0$ almost everywhere. If $V_{xx}(x_{0})=0$ holds at a point $x_0>x_e$, then 
\begin{align*}
\bigg[\beta V-\sup\limits_{\pi}\bigg(\frac{1}{2}\sigma^2\pi^2V_{xx}+\mu\pi V_x\bigg)-\sup\limits_{c\geqslant kx+l}\bigg(\frac{c^p}{p}-cV_x\bigg)-rxV_x\bigg](x_0)=-\infty,
\end{align*}
by virtue of $V_x(x_{0})>0$ and $\mu>0$, 
contradicting the HJB equation \eqref{HJB}. Thus, $V_{xx}(x)<0$ holds at any point where $V$ is twice differentiable.
\end{proof}

\begin{corollary}\label{Cor:bj}
We have that $V_{x}(x)$ is a decreasing, bijection mapping from $(x_{e},\infty)$ to $(0,\infty)$.
\end{corollary}
\begin{proof}
As $V_{x}(x)$ is continuous and decreasing, by \lemref{Lem:Vx0} and \lemref{Lem:Vx_up}, it suffices to prove that it is strictly decreasing.
If $V_{x}(x)$ is not strictly decreasing, then its monotonicity implies it is a constant on some interval; and consequently $V_{xx}(x)=0$ on that interval, contradicting \lemref{Lem:Vxx}. 
\end{proof}


\section{Dual transformation and $C^3$ smoothness }
\setcounter{equation}{0}
We now study the higher order differentiability of the value function in non-homogeneous case. The main tool is dual transformation, which turns the highly nonlinear HJB equation \eqref{HJB} into a semi-linear one that can be dealt with by classical ODE theory.

\par
Thanks to Lemma \ref{Lem:Vxx}, we can see that $V(x)$ satisfies the HJB equation \eqref{HJB} almost
everywhere (we will omit a.e. in the followings unless otherwise specified). Clearly 
\[ \pi^*(x)=\argmax_{\pi}\bigg(\frac{1}{2}\sigma^2\pi^2V_{xx}(x)+\mu\pi V_x(x)\bigg)=-\frac{\mu}{\sigma^2}\frac{V_x(x)}{V_{xx}(x)}, \]
so \eqref{HJB} can be rewritten as
\begin{align}\label{HJB1}
\beta V(x)+\frac{\mu^2}{2\sigma^2}\frac{V_x^2(x)}{V_{xx}(x)}-\sup\limits_{c\geqslant
kx+l}\bigg(\frac{c^p}{p}-cV_x\bigg)-rxV_x(x)=0, \quad x>x_e.
\end{align}
\par
Now, define the dual transformation of $V(x)$ (see Pham (2009)) as
\begin{align}\label{vV0}
v(y)=\sup\limits_{x> x_e}(V(x)-xy), \quad y>0.
\end{align}
By \corref{Cor:bj}, we can define a bijection mapping and its inverse mapping as 
\begin{align}\label{xy}
y=V_x(x)>0, \quad x=x(y):=V_x^{-1}(y)>x_{e}.
\end{align}
Clearly $x(y)$ is a maximizer for \eqref{vV0}, so 
\begin{align}\label{vV}
v(y)=V(x(y))-x(y)y, \quad y>0.
\end{align}
Moreover, by \corref{Cor:bj}, 
\begin{align}\label{vinfity}
v(+\infty)=\lim_{y\to+\infty}V(x(y))-x(y)y=-\infty.
\end{align}

\par
Differentiating \eqref{vV}, we get 
\begin{align}\label{vy=}
v_y(y)&=V_x(x(y))x'(y)-x'(y)y-x(y) =-x(y)<0, \\
\label{vyy=} v_{yy}(y) &=-x'(y)=-\frac{1}{V_{xx}(x(y))} >0\quad \text{a.e.}.
\end{align}
Moreover, by \corref{Cor:bj}, 
\begin{align}\label{vyinfity}
v_y(+\infty)=-\lim_{y\to+\infty}x(y)=-x_e.
\end{align}

Inserting \eqref{vy=} into \eqref{vV}, it follows
\begin{align}\label{v=V}
V(x(y))=v(y)-yv_y(y).
\end{align}
Applying \eqref{vy=}, \eqref{vyy=} and
\eqref{v=V}, \eqref{HJB1} then becomes
\begin{align}\label{v1}
\beta(v(y)-yv_y(y))-\frac{\mu^2}{2\sigma^2}y^2v_{yy}(y)-\sup\limits_{c\geqslant -kv_y(y)+l}\bigg(\frac{c^p}{p}-cy\bigg)+ryv_y(y)=0, \quad y>0.
\end{align}
This is a semi-linear ODE, which is degenerate at $y=0$.
\par
We now study the differentiability of the solution to \eqref{v1}.
Define
\begin{align*}
G(u, y)=\sup\limits_{c\geqslant u}\bigg(\frac{c^p}{p}-cy\bigg), \quad u>0, \;y>0. 
\end{align*}
Then \eqref{v1} can be written as 
\begin{align}\label{v}
\beta(v(y)-yv_y(y))-\frac{\mu^2}{2\sigma^2}y^2v_{yy}(y)-G(-kv_y(y)+l, y)+ryv_y(y)=0, \quad y>0.
\end{align}
\par 
Clearly, 
\begin{align*}
G(u, y)&=
\left\{
\begin{array}{ll}
\frac{1-p}{p}y^{\frac{p}{p-1}}, & \quad\hbox{if } y^{\frac{1}{p-1}}\geqslant u;\\[10pt]
\frac{u^p}{p}-uy, &\quad \hbox{if } y^{\frac{1}{p-1}}\leqslant u,
\end{array}
\right\}
\in C((0, \infty)\times(0, \infty)).
\end{align*}
Moreover,
\begin{align}\label{Gu}
G_u(u, y)&=
\left\{
\begin{array}{ll}
0, & \quad\hbox{if } y^{\frac{1}{p-1}}> u;\\
u^{p-1}-y, & \quad\hbox{if } y^{\frac{1}{p-1}}< u,
\end{array}
\right\}
=-(y-u^{p-1})^+, \\\label{Gy}
G_y(u, y)&=
\left\{
\begin{array}{ll}
-y^{\frac{1}{p-1}}, & \quad\hbox{if } y^{\frac{1}{p-1}}> u;\\
-u, & \quad\hbox{if } y^{\frac{1}{p-1}}< u,
\end{array}
\right\}
=-\max\{y^{\frac{1}{p-1}}, u\}.
\end{align}
It follows that $G(u, y)\in C^{1, 1}((0, \infty)\times(0, \infty))$. We see from \eqref{v} that $v(y)\in C^3((0, +\infty))$ (see Lady{\v{z}}enskaja, Solonnikov and Ural'ceva (1967)). Moreover, by \eqref{vy=} and \eqref{xy}, we can see that
\begin{align}\label{vy<-x0}
v_y(y)=-x(y)<-x_{e}=-\frac{l}{r-k}
\end{align}
for all $y>0$.

%
%
\subsection{Higher order differentiability}
We now state our main result. 
\begin{theorem}
The value function $V$ of the problem \eqref{V} belongs to $C^3((x_e, +\infty))$.
\end{theorem}
\begin{proof}
In the previous section, we find a solution $v\in C^3((0, +\infty))$ to the problem \eqref{v1} which is convex, decreasing and satisfies \eqref{vinfity} and \eqref{vyinfity}.
Let 
\[\varphi(x)=\inf\limits_{y>0}(v(y)+xy), \quad y>0.\]
Then 
\[v(y)=\sup\limits_{x> x_e}(\varphi(x)-xy), \quad y>0.\] 
Similar to \eqref{vy=} and \eqref{vyy=}, we can show 
\[ \varphi_{xx}(x) v_{yy}(v_y^{-1}(-x))=-1, \quad x>x_e. \]
Note that $v\in C^3((0, +\infty))$, so $\varphi\in C^3((x_e, +\infty))$ is an immediate consequence of the following claim
\[ v_{yy}(y)>0, \quad y>0. \]
To show this, 
differentiating \eqref{v}, we have
\[ (r-\beta) yv_{yy}-\frac{\mu^2}{2\sigma^2}(2yv_{yy}+y^2v_{yyy})-G_u(-kv_y+l, y)(-kv_{yy})-G_y(-kv_y+l, y)+rv_y=0. \]
Substituting \eqref{Gu} and \eqref{Gy} we get
\begin{multline}\label{vy}
(r-\beta) yv_{yy}-\frac{\mu^2}{2\sigma^2}(2yv_{yy}+y^2v_{yyy})+(y-(-kv_y+l)^{p-1})^+(-kv_{yy})\\
+\max\{y^{\frac{1}{p-1}}, -kv_y+l\}+rv_y=0.
\end{multline}
By the definition \eqref{vV0}, $v(y)$ is convex, so $v_{yy}(y)\geqslant 0$ for all $y>0$. 
We next prove $v_{yy}(y)>0$ by contradiction. 
\par
Suppose 
$v_{yy}(y_0)=0$ for some $y_0>0$. Then $v_{yy}$ attains its minimum value 0 at $y_0$,
so $v_{yyy}(y_0)=0$. Substituting them into \eqref{vy}, we obtain
\begin{align}\label{vyy=0}
\max\{y_0^{\frac{1}{p-1}}, -kv_y(y_0)+l\}+rv_y(y_0)=0.
\end{align}
If $-kv_y(y_0)+l\geqslant y_0^{\frac{1}{p-1}}$, then \eqref{vyy=0} becomes
\[ -kv_y(y_0)+l+rv_y(y_0)=0,\]
which implies $v_y(y_0)=-x_e$, contradicting \eqref{vy<-x0}.
Otherwise, $y_0^{\frac{1}{p-1}}>-kv_y(y_0)+l$ so that there exists a neighborhood of $y_0$, denoted by $\Omega$, such that
$y^{\frac{1}{p-1}}>-kv_y(y)+l$ for $y\in\Omega$. By \eqref{vy},
\begin{align*}
(r-\beta) yv_{yy}-\frac{\mu^2}{2\sigma^2}(2yv_{yy}+y^2v_{yyy})+rv_y=-y^{\frac{1}{p-1}} \quad\hbox{in}\quad\Omega.
\end{align*}
Differentiating it, we have
\begin{align*}
-\frac{\mu^2}{2\sigma^2}y^2v_{yyyy}(y)+\bigg(r-\frac{2\mu^2}{\sigma^2}-\beta\bigg)yv_{yyy}(y)+2rv_{yy}(y)=\frac{1}{1-p}y^{\frac{2-p}{p-1}}+\bigg(\beta+\frac{\mu^2}{\sigma^2}\bigg)v_{yy}(y)>0
\end{align*}
in $\Omega$. By the strong maximum principle and together with $v_{yy}\geqslant 0$ on the boundary points of $\Omega$, we obtain $v_{yy}> 0$ in $\Omega$, contradicting $v_{yy}(y_0)=0$.

The above argument shows that $\varphi$ is a classical solution to the HJB equation \eqref{HJB} in the class of increasing, concave $C^3$ functions. Therefore $V=\varphi$ by Lemma \ref{varthm}. 
\end{proof}


\section{Optimal strategy in non-homogeneous case}
\setcounter{equation}{0}

In this section we study the optimal consumption-investment strategy for the problem \eqref{V} in the non-homogeneous case $0<k<r$, $l>0$.

To this end, we divide the whole state space $(x_e,+\infty)$ into an unconstrained consumption region $\cal U$ and a constrained consumption region $\cal C$ as follows
\begin{align*}
{\cal U}&:=\{x>x_e: V_x(x)^{\frac{1}{p-1}}> kx+l\};\\
{\cal C} &:=\{x>x_e: V_x(x)^{\frac{1}{p-1}}\leqslant kx+l\}.
\end{align*}
One should consume optimally in ${\cal U}$, and consume as few as possible in ${\cal C}$. 
Precisely the optimal consumption-investment strategy for the problem \eqref{V} is stated as follows. 
\begin{theorem}[Optimal consumption-investment strategy]
The optimal investment-investment strategy for the problem \eqref{V} is given by the feedback law 
as
\begin{align*}
c^{*}(x)&=\begin{cases}
V_x(x)^{\frac{1}{p-1}},\quad x\in{\cal U};\\
k x+l,\quad x\in{\cal C},
\end{cases} \quad\pi^*(x)=-\frac{\mu}{\sigma^2}\frac{V_x(x)}{V_{xx}(x)}.
\end{align*}
\end{theorem}
\begin{proof}
This comes from the HJB equation \eqref{HJB} and Lemma \ref{varthm}. 
\end{proof}

In the rest of the paper, we focus on the properties of the two regions ${\cal C}$ and ${\cal U}$.

Our first observation is that, intuitively speaking, one is expected to consume as few as possible when her financial situation is fairly bad. The following result confirms this fact.
\begin{lemma}[Optimal consumption strategy in bad financial situation]\label{Lem:x1}
There exists $x_1>x_e$ such that
\begin{align}\label{x1}
V_x(x)^{\frac{1}{p-1}}< kx+l, \quad x_e< x< x_1;
\end{align}
and thus $(x_e, x_1)\subset{\cal C}$.
\end{lemma}
\begin{proof}
By Lemma \ref{Lem:Vx0}, $\lim\limits_{x\rightarrow x_e+}V_x(x)^{\frac{1}{p-1}}=0$. So the claim follows as $l>0$.
\end{proof}

When one's financial situation is fairly good, we will show that the discount factor $\beta$ plays a critical role in determining the optimal consumption strategy. 
If the discount factor is small, then the investor cares more about long-term than short-term consumption. As a result, she should consume as few as possible and save the reminder for future investment. This consumption habit is reversed when the discount factor is large, she should consume optimally when her wealth is above a threshold, and consume as few as possible when her wealth is below the threshold.

\subsection{Optimal consumption with small discount factor}

Recall that
\[ \kappa=\frac{\beta-p\left(\frac{\mu^2}{2\sigma^2(1-p)}+r\right)}{1-p}\]
is strictly increasing in $\beta$, so its value reflects whether the discount factor is large or small.

In this part, we study the case with a small discount factor. 

The following result states that if one is more concerned with long-term than short-term consumption, then she should, regardless of her financial condition, always consume as few as possible. 
\begin{theorem}
If $0<\kappa<k$, then ${\cal C}=(x_e, +\infty)$ and ${\cal U}=\emptyset$.
\end{theorem}
\begin{proof}
We first show that there exists $x_3>x_e$ such that 
\begin{align}\label{x3}
V_x(x)^{\frac{1}{p-1}}< kx+l, \quad x> x_3,
\end{align}
which indicates that ${\cal U}$ is a bounded set. 
In fact, by \eqref{Vk}, we have $V_k(x)=\frac{k^p}{p(\kappa(1-p)+kp)}x^p=\frac{1}{p}h^{p-1}x^p$, where
\[ h=\frac{k^\frac{p}{p-1}}{(\kappa(1-p)+kp)^\frac{1}{p-1}}<k. \]
By \eqref{V_up} and \eqref{V_downup}, we have
\begin{align}\label{V_bound}
\frac{1}{p}h^{p-1}(x-x_e)^p\leqslant V(x) \leqslant \frac{1}{p}h^{p-1}x^p.
\end{align}
The concavity of $V(x)$ implies 
\[ V_x(x)\geqslant \frac{V(x+y)-V(x)}{y} , \quad y>0. \]
Together with \eqref{V_bound}, we infer that 
\[ V_x(x)\geqslant \frac{h^{p-1}[(x+y-x_e)^p-x^p]}{py} , \quad y>0. \]
Setting $y=\varepsilon x$ in above, it follows 
\[ \frac{V_x(x)}{x^{p-1}}\geqslant \frac{h^{p-1}[(1+\varepsilon-\frac{x_e}{x})^p-1]}{p\varepsilon}, \]
so 
\[ \liminf\limits_{x\rightarrow+\infty}\frac{V_x(x)}{x^{p-1}}\geqslant \frac{h^{p-1}[(1+\varepsilon)^p-1]}{p\varepsilon}\rightarrow h^{p-1}, \quad \text{ as}\quad \varepsilon \rightarrow0. \]
Hence,
\[ \limsup\limits_{x\rightarrow+\infty}\frac{V_x(x)^{\frac{1}{p-1}}}{kx+l} =\frac{1}{k}\bigg(\liminf\limits_{x\rightarrow+\infty}\frac{V_x(x)}{x^{p-1}}\bigg)^{\frac{1}{p-1}} \leqslant \frac{h}{k}<1. \]
This confirms \eqref{x3}. 

To prove ${\cal C}=(x_e, +\infty)$, it is sufficient to prove 
\begin{align}\label{kx+l}
V_x(x)^{\frac{1}{p-1}}\leqslant kx+l, \quad x>x_e.
\end{align} 
From the dual-problem viewpoint, we need only prove
\begin{align}\label{-kv_y+l>}
v_y(y)\leqslant-\frac{1}{k}y^{\frac{1}{p-1}}+\frac{l}{k}, \quad y>0.
\end{align}
By Lemma \ref{Lem:x1} and \eqref{x3}, for extremely small and extremely large $y$, the above holds true. 

Per absurdum, suppose there exists a bounded open interval $D\subseteq (0, +\infty)$ such that
\begin{align}\label{-kv_y+l<}
\left\{
\begin{array}{l}
v_y(y)>-\frac{1}{k}y^{\frac{1}{p-1}}+\frac{l}{k}, \quad y\in D, \\
v_y(y)=-\frac{1}{k}y^{\frac{1}{p-1}}+\frac{l}{k}, \quad y\in \partial D.
\end{array}
\right.
\end{align}
Then $D$ is bounded away from zero. Owing to \eqref{vy}, we have
\[ (r-\beta) yv_{yy}-\frac{\mu^2}{2\sigma^2}(2yv_{yy}+y^2v_{yyy})+y^{\frac{1}{p-1}}+rv_y=0, \quad y\in D. \]
Denote
\begin{align}\label{operatorL}
{\cal L}(u)=-\frac{\mu^2}{2\sigma^2}y^2u_{yy}+\bigg(r-\beta-\frac{\mu^2}{\sigma^2}\bigg) yu_{y}+ru,
\end{align}
then
${\cal L}(v_y)=-y^{\frac{1}{p-1}}$ in $D$.
\par
On the other hand, in $D$,
\begin{align}
&\hspace{1.3em}{\cal L}\bigg(-\frac{1}{k}y^{\frac{1}{p-1}}+\frac{l}{k}\bigg)\nonumber\\
&= -\frac{\mu^2}{2\sigma^2}\bigg(-\frac{1}{k}\frac{1}{p-1}\bigg(\frac{1}{p-1}-1\bigg)\bigg)y^{\frac{1}{p-1}}+\bigg(r-\beta-\frac{\mu^2}{\sigma^2}\bigg)\bigg(-\frac{1}{k}\frac{1}{p-1}\bigg)y^{\frac{1}{p-1}}\nonumber\\
&\hspace{1.3em}+r\bigg(-\frac{1}{k}y^{\frac{1}{p-1}}+\frac{l}{k}\bigg)\nonumber\\
&= \bigg[\frac{\mu^2}{2\sigma^2}\bigg(\frac{1}{1-p}+1\bigg)+\bigg(r-\beta-\frac{\mu^2}{\sigma^2}\bigg)-r(1-p)\bigg]\frac{y^{\frac{1}{p-1}}}{k(1-p)}+\frac{rl}{k}\nonumber\\
&= \bigg[p\bigg(\frac{\mu^2}{2\sigma^2(1-p)}+r\bigg)-\beta\bigg]\frac{y^{\frac{1}{p-1}}}{k(1-p)}+\frac{rl}{k}\nonumber\\\label{Lu} &= \frac{-\kappa}{k}y^{\frac{1}{p-1}}+\frac{rl}{k}. 
\end{align}
Recall $0<\kappa<k$, so we have
\begin{align*}
{\cal L}\bigg(-\frac{1}{k}y^{\frac{1}{p-1}}+\frac{l}{k}\bigg)\geqslant -y^{\frac{1}{p-1}}+\frac{rl}{k}\geqslant-y^{\frac{1}{p-1}}={\cal L}(v_y).
\end{align*}
Using the comparison principle for ODE, we obtain $v_y(y)\leqslant-\frac{1}{k}y^{\frac{1}{p-1}}+\frac{l}{k}$ in
$D$, contradicting \eqref{-kv_y+l<}.
\end{proof}

\subsection{Optimal consumption with large discount factor}
When the discount factor is fairly large, the optimal consumption strategy is no more always consuming at the minimal rate. In fact, the strategy is state-dependent. We will show that there exists a critical threshold such that the investor should consume optimally when her wealth is above the threshold, and consume as few as possible below it. The result is stated as follows.
\begin{theorem}
If $\kappa \geqslant r$, then ${\cal C}=(x_e,x^*]$ and ${\cal U}=(x^*,+\infty)$ for some $$x^*\in \left(x_e, \quad\frac{x_e+(1-p)(\frac{k}{\kappa})^{-p}\frac{l}{\kappa}}{1-(\frac{k}{\kappa})^{1-p}}\right).$$
\end{theorem}
\begin{proof}
Because $\kappa\geq r>k$, owing to \eqref{Vx_up},
\[
\liminf\limits_{x\rightarrow+\infty}\frac{V_x(x)^{\frac{1}{p-1}}}{kx+l} \geqslant
\liminf\limits_{x\rightarrow+\infty}\frac{\left(\frac{\kappa^{p-1}x^p}{x-x_e}\right)^{\frac{1}{p-1}}}{kx+l}
=\frac{\kappa}{k}>1.
\]
So there exists $x_2\in (x_e, +\infty)$ such that
\begin{align}\label{x2}
V_x(x)^{\frac{1}{p-1}}> kx+l, \quad x> x_2, 
\end{align} 
which implies ${\cal C}$ is a bounded set. 
Define \[x^*=\inf\{x>x_e\mid V_x(x)^{\frac{1}{p-1}}\geqslant kx+l\}.\] 
We get $x^*\in (x_e, +\infty)$ from Lemma \ref{Lem:x1} and \eqref{x2}.
We now show 
\begin{align}\label{geq}
V_x(x)^{\frac{1}{p-1}}&\geqslant kx+l, \quad x>x^*.
\end{align}
Suppose this is not the case. 
Because ${\cal C}$ is a bounded set and $V_x$ is continuous, 
there exist $a$, $b$ such that $x_{e}<a<b<\infty$ and 
\begin{gather}\label{Vxab1}
V_x(x)^{\frac{1}{p-1}}\leqslant kx+l, \quad x\in (a, b), \\\label{Vxab2}
V_x(a)^{\frac{1}{p-1}}=ka+l, \quad V_x(b)^{\frac{1}{p-1}}=kb+l.
\end{gather}
By the dual relationship, we have
\begin{gather}\label{vyab1}
v_y(y)\leqslant-\frac{1}{k}y^{\frac{1}{p-1}}+\frac{l}{k}, \quad y\in
(y_b, y_a), \\\label{vyab2}
v_y(y_a)=-\frac{1}{k}y_a^{\frac{1}{p-1}}+\frac{l}{k}, \quad v_y(y_b)=-\frac{1}{k}y_b^{\frac{1}{p-1}}+\frac{l}{k}, \nonumber
\end{gather}
where $y_a=V_x(a)$, $y_b=V_x(b)$. By \eqref{vy}, when $y\in (y_b, y_a)$,
\[ (r-\beta) yv_{yy}-\frac{\mu^2}{2\sigma^2}(2yv_{yy}+y^2v_{yyy})+(y-(-kv_y+l)^{p-1})^+(-kv_{yy})-kv_y+l+rv_y=0. \]
Note that $v_{yy}>0$, thus
\begin{align}\label{Lvy}
{\cal L}(v_y)-kv_y+l=(y-(-kv_y+l)^{p-1})^+(kv_{yy})\geqslant 0, 
\end{align}
where ${\cal L}(u)$ is defined by \eqref{operatorL}.
On the other hand, by \eqref{Lu} and $\kappa\geqslant r$, we have
\begin{align}\label{Lw}
{\cal L}\bigg(-\frac{1}{k}y^{\frac{1}{p-1}}+\frac{l}{k}\bigg) = \frac{-\kappa}{k}y^{\frac{1}{p-1}}+\frac{rl}{k}
\leqslant r\bigg(-\frac{1}{k}y^{\frac{1}{p-1}}+\frac{l}{k}\bigg).
\end{align}
Denote
\begin{align*}
{\cal F}(u(y))={\cal L}(u(y))-ru(y)=-\frac{\mu^2}{2\sigma^2}y^2u_{yy}+\bigg(r-\beta-\frac{\mu^2}{\sigma^2}\bigg) yu_{y},
\end{align*}
then, by \eqref{Lvy}, \eqref{Lw} and \eqref{vy<-x0},
\[ {\cal F}\bigg(v_y(y)-\big(-\tfrac{1}{k}y^{\frac{1}{p-1}}+\tfrac{l}{k}\big)\bigg)\geqslant (r-k)(-v_y(y))-l>(r-k)x_e-l=0, \quad y\in(y_b, y_a).\]
By the strong maximum principle, we have
$v_y(y)-\bigg(-\frac{1}{k}y^{\frac{1}{p-1}}+\frac{l}{k}\bigg)>0$ in $(y_b, y_a)$, contradicting \eqref{vyab1}.
Hence \eqref{geq} is proved. 

We now show the inequality in \eqref{geq} is strict.
Denote $y^*=v_y^{-1}(-x^*)=V_x(x^*)$.
It suffices to prove
\begin{align}\label{y^*}
v_y(y)>-\frac{1}{k}y^{\frac{1}{p-1}}+\frac{l}{k}, \quad y\in (0, y^*).
\end{align}
By \eqref{vy}, \eqref{operatorL} and \eqref{geq}, we know
\begin{align*}
{\cal L}(v_y)
&= (r-\beta-\frac{\mu^2}{\sigma^2}) yv_{yy}-\frac{\mu^2}{2\sigma^2}y^2v_{yyy}+rv_y\\
&= (y-(-kv_y+l)^{p-1})^+kv_{yy}-\max\{y^{\frac{1}{p-1}}, -kv_y+l\}\\
&\geqslant -\max\{y^{\frac{1}{p-1}}, -kv_y+l\}\\
&= -y^{\frac{1}{p-1}}, \quad y\in (0, y^*).
\end{align*}
Together with \eqref{Lw}, \eqref{vy<-x0} and \eqref{geq}, we have
\begin{align*}
{\cal L}\bigg(v_y(y)-\bigg(-\frac{1}{k}y^{\frac{1}{p-1}}+\frac{l}{k}\bigg)\bigg)
&\geqslant
\frac{r-k}{k}y^{\frac{1}{p-1}}-\frac{rl}{k}\\
&\geqslant \frac{r-k}{k}(-kv_y+l)-\frac{rl}{k}\\
&= (r-k)(-v_y)-l\\
&> (r-k)x_e-l=0, \quad y\in (0, y^*).
\end{align*}
Using the strong comparison principle, we deduce \eqref{y^*}.

It is only left to prove the upper bound for $x^{*}$. 
Since $x^*\in {\cal C}$, we have $kx^*+l\geqslant V_x^{\frac{1}{p-1}}(x^*)$. Together with \eqref{Vx_up}, we get
\[ (kx^*+l)^{p-1}\leqslant V_x(x^*)\leqslant \frac{\kappa^{p-1}(x^*)^p}{x^*-x_e}, \]
thus,
\[ \left(\frac{k}{\kappa}+\frac{l}{\kappa x^*}\right)^{1-p}\geqslant 1-\frac{x_e}{x^*}. \]
For any $a,t>0$, by the mean-value theorem, there exists $b\in(0,t)$ such that 
\[ (a+t)^{1-p}-a^{1-p}=(1-p)(a+b)^{-p}t<(1-p)a^{-p}t.\] 
Hence, 
\[ (1-p)\left(\frac{k}{\kappa}\right)^{-p}\frac{l}{\kappa x^*} >\left(\frac{k}{\kappa}+\frac{l}{\kappa x^*}\right)^{1-p}-\left(\frac{k}{\kappa}\right)^{1-p}\geqslant 1-\frac{x_e}{x^*}-\left(\frac{k}{\kappa}\right)^{1-p}, \]
yielding 
\[ x^*<\frac{x_e+(1-p)(\frac{k}{\kappa})^{-p}\frac{l}{\kappa}}{1-(\frac{k}{\kappa})^{1-p}}. \]
This completes the proof.
\end{proof}


\section{Concluding remarks}
\setcounter{equation}{0}
In the most interesting case $k,l>0$, we proved that if $0<\kappa< k$, then ${\cal C}=(x_e, +\infty)$ and ${\cal U}=\emptyset$, i.e., the optimal consumption rate is always the lower bound; if $\kappa> k$, both ${\cal C}$ and ${\cal U}$ are not empty; furthermore, they are both intervals specifically when $\kappa\geqslant r$. However, in the scenario $0<k\leq \kappa<r$, whether they are connected regions remains unknown.

\begin{appendices}

\section{The case with state-independent constraint $k=0$, $l>0$}
\setcounter{equation}{0}

Let us consider the scenario of $\kappa>k=0$, $l>0$. Corresponding to the equation \eqref{v1}, we have
\begin{align*}
\beta(v(y)-yv_y(y))-\frac{\mu^2}{2\sigma^2}y^2v_{yy}(y)-\sup\limits_{c\geqslant l}\bigg(\frac{c^p}{p}-cy\bigg)
+ryv_y(y)=0, \quad y>0,
\end{align*}
i.e.,
\begin{align}\label{v2}
-\frac{\mu^2}{2\sigma^2}y^2v_{yy}(y)+(r-\beta)yv_y(y)+\beta v(y)=
\left\{
\begin{array}{ll}
\frac{1-p}{p}y^{\frac{p}{p-1}}, & 0<y\leqslant l^{p-1}; \\ [3mm]
\frac{l^p}{p}-ly, & y\geqslant l^{p-1}.
\end{array}
\right.
\end{align}
Denote
\[ f(\lambda)=-\frac{\mu^2}{2\sigma^2}\lambda(\lambda-1)+(r-\beta)\lambda+\beta, \]
and let $\lambda_1<0$, $\lambda_2>1$ be the two roots of $f$. Then, the general solution of
\eqref{v2} is
\begin{align*}
v(y)=
\left\{
\begin{array}{ll}
C_1y^{\lambda_1}+C_2y^{\lambda_2}+\frac{1-p}{p\kappa}y^{\frac{p}{p-1}}, & 0<y\leqslant l^{p-1}; \\ [3mm]
D_1y^{\lambda_1}+D_2y^{\lambda_2}+\frac{l^p}{\beta p}-\frac{l}{r}y, & y\geqslant l^{p-1}.
\end{array}
\right.
\end{align*}
Owing to \eqref{V_up} and \eqref{V_downup},
\[ \frac{\kappa^{p-1}(x-x_e)^p}{p}\leqslant V(x)\leqslant \frac{\kappa^{p-1}x^p}{p}. \]
As such,
\[ v(y)=\max\limits_{x>x_e}(V(x)-xy)
\left\{
\begin{array}{ll}
\geqslant \max\limits_{x>x_e}\bigg(\frac{\kappa^{p-1}(x-x_e)^p}{p}-xy\bigg)=\frac{1-p}{p\kappa}y^{\frac{p}{p-1}}-x_ey; \\
\leqslant \max\limits_{x>0}\bigg(\frac{\kappa^{p-1}x^p}{p}-xy\bigg)=\frac{1-p}{p\kappa}y^{\frac{p}{p-1}}.
\end{array}
\right.
\]
Because $f\left(\frac{p}{p-1}\right)=\kappa>0$, 
$\lambda_1<\frac{p}{p-1}<\lambda_2$. 
From the above estimates for $v$, we see $C_1=0$ and $D_2=0$. Therefore,
\begin{align*} 
v(y)=
\left\{
\begin{array}{ll}
Cy^{\lambda_2}+\frac{1-p}{p\kappa}y^{\frac{p}{p-1}}, & 0<y\leqslant l^{p-1}; \\ [3mm]
Dy^{\lambda_1}+\frac{l^p}{\beta p}-\frac{l}{r}y, & y\geqslant l^{p-1},
\end{array}
\right.
\end{align*}
and consequently,
\begin{align*} 
v_y(y)=
\left\{
\begin{array}{ll}
C\lambda_2y^{\lambda_2-1}-\frac{1}{\kappa}y^{\frac{1}{p-1}}, & 0<y< l^{p-1}; \\ [3mm]
D\lambda_1y^{\lambda_1-1}-\frac{l}{r}, & y> l^{p-1}.
\end{array}
\right.
\end{align*}
Using the continuity condition of $v$ and $v_{y}$ at $l^{p-1}$, we obtain
\begin{align*}
\left\{
\begin{array}{ll}
Cl^{\lambda_2(p-1)}+\frac{1-p}{p\kappa}l^p=Dl^{\lambda_1(p-1)}+\frac{l^p}{\beta p}-\frac{l^p}{r}; \\ [3mm]
C\lambda_2l^{(\lambda_2-1)(p-1)}-\frac{1}{\kappa}l
=D\lambda_1l^{(\lambda_1-1)(p-1)}-\frac{l}{r},
\end{array}
\right.
\end{align*}
or equivalently, 
\begin{align*}
\left\{
\begin{array}{ll}
C=\bigg(\frac{\lambda_1-1}{r}-\frac{\lambda_1}{\beta p}+\frac{1-p}{p}\frac{\lambda_1}{\kappa}+\frac{1}{\kappa}\bigg)
\frac{l^{p+\lambda_2(1-p)}}{\lambda_2-\lambda_1}; \\ [3mm]
D=\bigg(\frac{\lambda_2-1}{r}-\frac{\lambda_2}{\beta p}+\frac{1-p}{p}\frac{\lambda_2}{\kappa}+\frac{1}{\kappa}\bigg)
\frac{l^{p+\lambda_1(1-p)}}{\lambda_2-\lambda_1}.
\end{array}
\right.
\end{align*}
In this case, the free boundary point is
\begin{align*}
x^*&=-v_y(y^*)=-v_y(l^{p-1}) = D\lambda_1l^{(\lambda_1-1)(p-1)}-\frac{l}{r}\\
&=\bigg(\frac{\lambda_2-1}{r}-\frac{\lambda_2}{\beta p}+\frac{1-p}{p}\frac{\lambda_2}{\kappa}
+\frac{1}{\kappa}\bigg)
\frac{l^{p+\lambda_1(1-p)}}{\lambda_2-\lambda_1}\lambda_1l^{(\lambda_1-1)(p-1)}-\frac{l}{r}\\
&= \bigg(\frac{\lambda_2-1}{r}-\frac{\lambda_2}{\beta p}+\frac{1-p}{p}\frac{\lambda_2}{\kappa}
+\frac{1}{\kappa}\bigg)
\frac{\lambda_1 l}{\lambda_2-\lambda_1}-\frac{l}{r}.
\end{align*}
And the two regions are ${\cal C}=(x_e,x^*]$ and ${\cal U}=(x^*,+\infty)$.

\end{appendices}

\end{document}